\documentclass[3p,review]{elsarticle}

\usepackage{lineno,hyperref,graphicx,epstopdf}
\usepackage{times}
\usepackage{type1cm}
\fontsize{12pt}{18pt}

\usepackage{braket}
\usepackage{amsmath}
\usepackage{amsfonts}
\usepackage{amssymb}
\usepackage{amsthm}

\theoremstyle{plain}
\newtheorem{theorem}{Theorem}

\newtheorem{lemma}{Lemma}
\newtheorem{proposition}{Proposition}
\newtheorem{corollary}{Corollary}

\theoremstyle{definition}
\newtheorem{definition}{Definition}

\newtheorem{remark}{Remark}

\journal{}

\newcommand{\totalref}[2]{{#1 \ref{#2}}}
\date{}

\begin{document}
\begin{frontmatter}
\title{State complexity of one-way quantum finite automata together with classical states}

\author{Ligang Xiao}
\author{Daowen Qiu\corref{one}}

\cortext[one]{Corresponding author (D. Qiu). {\it E-mail addresses:} issqdw@mail.sysu.edu.cn (D. Qiu)}

\address{
   Institute of Quantum Computing and Computer Theory, School of Computer Science and Engineering, Sun Yat-sen University, Guangzhou 510006, China }

\begin{abstract}
\emph{One-way quantum finite automata together with classical states} (1QFAC) proposed in [Journal of Computer and
  System Sciences 81~(2) (2015) 359--375] is a new \emph{ one-way quantum finite automata} (1QFA) model that integrates \emph{quantum finite automata} (QFA) and \emph{deterministic finite automata} (DFA). This model uses classical states to control the evolution and measurement of quantum states. As a quantum-classical hybrid model, 1QFAC recognize all regular languages. It was shown that the state complexity of 1QFAC for some languages is essentially superior to that of DFA and other 1QFA. However, the relationships and balances between quantum states and classical states are still not clear in 1QFAC, for example, how to reduce a quantum state by adding classical states, and vice versa. In this paper, our goal is to clarify state complexity problems for 1QFAC. We obtain the following results: (1) We optimize the bound given by Qiu et al. that characterizes the relationship between quantum basis state number and classical state number of 1QFAC as well as the state number of its corresponding minimal DFA for recognizing any given regular language. (2) We give an upper bound showing that how many classical states are needed if the quantum basis states of 1QFAC are reduced without changing its recognition ability. (3) We give a lower bound of the classical state number of 1QFAC for recognizing any given regular language, and the lower bound is exact if the given language is finite. (4) We show that 1QFAC are exponentially more succinct than DFA and \emph{probabilistic finite automata} (PFA) for recognizing some regular languages that can not be recognized by measure-once 1QFA (MO-1QFA), measure-many 1QFA (MM-1QFA) or multi-letter 1QFA. (5) We reveal essential relationships between 1QFAC, MO-1QFA and multi-letter 1QFA, and induce a result regarding a quantitative relationship between the state number of multi-letter 1QFA and DFA.
\end{abstract}


\end{frontmatter}


\section{Introduction}\label{introduction}
Quantum computing has shown its great advantages in some aspects, such as Shor's factoring algorithm \cite{shor1994algorithms} 
  and Grover's search algorithm \cite{grover1996a}. Shor's  algorithm is exponentially faster than the corresponding known classical algorithms in factoring a large number and could be used to crack the RSA cryptosystem. Grover's search algorithm has square root acceleration compared to classical algorithms in finding a target of an unstructured set. In a way, these examples regard the power of quantum Turing machines. Nowadays it is still difficult to create large-scale universal quantum computers. Thus, another orientation is to consider more restricted theoretical models of quantum computers, such as \emph{one-way quantum finite automata} (1QFA).
	
	1QFA are restricted theoretical models of quantum computers with finite memory and their tape heads only move one cell to the right at each step. There are many kinds of 1QFA, the most basic of which are \emph{measure-once} 1QFA (MO-1QFA) proposed by Moore and Crutchfield \cite{moore2000quantum} and \emph{measure-many} 1QFA (MM-1QFA) proposed by Kondacs and Watrous \cite{kondacs1997on}. Others include  Latvian QFA  \cite{ABG06},  1QFA with control language \cite{bertoni2003quantum},    Ancilla QFA \cite{Pas00},  \emph{multi-letter} 1QFA \cite{belovs2007multi},  one-way quantum finite automata together with classical states (1QFAC) \cite{QIU20151QFAC}, etc. In MO-1QFA, each evolution depends on the last letter received and there is only one measurement performed after reading the input string. Unlike MO-1QFA, measurement in MM-1QFA is performed after reading each symbol. A $k$-letter 1QFA is ``measure-once" like MO-1QFA, but each of its evolution depends on the last $k$ letters received currently. Indeed, $k+1$-letter 1QFA have stronger ability of language recognizability than $k$-letter 1QFA \cite{QY09}, and
 a 1-letter 1QFA is an MO-1QFA exactly.

	MO-1QFA, MM-1QFA and multi-letter 1QFA accept proper subsets of regular languages with bounded error. MM-1QFA can accept more languages than MO-1QFA with bounded error \cite{golovkins2002probabilistic}, and multi-letter 1QFA can accept some regular languages not accepted by MM-1QFA \cite{bertoni2003quantum,characterizations}.  For more models of 1QFA, we can refer to \cite{bell2021on,bhatia2019quantum,nishimura2009an,say2014quantum,yakaryilmaz2010succinctness,yamakami2014one}. Besides 1QFA, there are many other important quantum automata such as \emph{two-way QFA} proposed and studied by  Kondacs, Watrous, and Ambainis \cite{AW02,kondacs1997on}.

	Recently, Qiu et al. proposed a new model of one-way finite automata that integrates 1QFA and DFA, namely, 1QFA \emph{together with classical states} (1QFAC) \cite{QIU20151QFAC}. We describe the computing procedure roughly. At the beginning, 1QFAC is in an initial classical state and an initial quantum state. After reading each input symbol, the current classical state together with current input symbol assigns a unitary transformation to act on the current quantum state, and the current classical state is updated by means of classical transformations. When the last input symbol has been scanned, a measurement  in terms of the last classical state is assigned to perform on the final quantum state, producing a result of accepting or rejecting the input string. Qiu et al. \cite{QIU20151QFAC} have proved that 1QFAC only accept all regular languages, and investigated  other basic problems of 1QFAC. They first gave the bound that uses state number of 1QFAC for recognizing any given regular language to bound the state number of its corresponding minimal DFA. Then, they showed that 1QFAC are exponentially more concise than DFA in some languages that can not be recognized by MO-1QFA or  MM-1QFA or can not be recognized by multi-letter 1QFA. Moreover, they solved the equivalence problem and minimization problem of 1QFAC.

	However, the bound mentioned above is not tight enough and it remains a number of state complexity problems of 1QFAC to be solved. More specifically, there are three pending problems:  given any regular language $L$ and given a 1QFAC recognizing $L$, if another 1QFAC recognizes $L$ with fewer quantum basis states, how many classical states does it need to be added? What is the lower bound of the classical state number of 1QFAC for recognizing any given regular language? Whether the state complexity of 1QFAC for some languages is superior to that of DFA, PFA and other 1QFA? So in this paper, our goals are to solve these problems. We study these state complexity problems of 1QFAC and reveal essential relationships between 1QFAC, MO-1QFA and multi-letter 1QFA.
It is worth mentioning that the state complexity between 1QFA and DFA has been studied (for example, \cite{bianchi2014size,li2016lower} and the references therein).


	The remainder of the paper is organized as follows. In Section \ref{sec:preliminaries}, we introduce related notations and recall the definitions of a number of one-way finite automata, including DFA, PFA, MO-1QFA, MM-1QFA, multi-letter 1QFA and 1QFAC. Then in Section \ref{sec:state-complexity}, we study state complexity problems of 1QFAC. We optimize the quantitative relationship between the state number of 1QFAC and DFA that was given in \cite{QIU20151QFAC}, clarify the trade-offs between quantum basis states and classical states and give a lower bound of the classical state number of 1QFAC for recognizing any given regular language. Afterwards in Section \ref{sec:state_succinctness}, we give an example showing that 1QFAC are exponentially more succinct than DFA and PFA in given regular languages but none of MO-1QFA, MM-1QFA and multi-letter 1QFA can recognize them. Finally in Section \ref{sec:simulation}, we reveal essential relationships between 1QFAC, MO-1QFA and multi-letter 1QFA, and induce a result regarding a quantitative relationship between the state number of multi-letter 1QFA and DFA.


\section{Preliminaries}\label{sec:preliminaries}
	In this section, we review related notations in quantum computing and recall one-way QFA. For the details, we can refer to \cite{bhatia2019quantum, nielsen2000quantum}.

\subsection{Basic notations in quantum computing}
	We recall some notations that are useful in this paper, for the details we can refer to \cite{nielsen2000quantum}. We denote by $\mathbb{R},\mathbb{C},\mathbb{N},\mathbb{Z}$ and $\mathbb{Z}^+$ the sets of real number, complex number, natural number, integer and positive integer, respectively. Let $\mathbb{F}$ be a number field. We denote by $\mathbb{F}^{m\times n}$ the set of $m\times n$ matrices whose entries are in $\mathbb{F}$, and we write $\mathbb{F}^n$ instead of $\mathbb{F}^{n\times 1}$. We denote by $I$ and $O$ the identity matrix and zero matrix, respectively. Given $M\in\mathbb{C}^{n\times m}$, we denote by $M^*,M^T,M^{\dag}$ its conjugate, transpose and adjoint, respectively, where $M^{\dag}=(M^T)^*$. $M\in \mathbb{C}^{n\times n}$ is said to be a unitary matrix if $MM^{\dag}=M^{\dag}M=I$. $P\in \mathbb{C}^{n\times n}$ is said to be a projective matrix if $P=P^{\dag}$ and $P^2=P$.
    For any $A\in \mathbb{C}^{n_1\times m_1}$ and any $B\in \mathbb{C}^{n_2\times m_2}$, their \emph{tensor product} $A\otimes B\in \mathbb{C}^{(n_1n_2)\times (m_1m_2)}$ is defined as
$$A\otimes B=\begin{bmatrix}
A_{11}B & \cdots & A_{1n}B \\
\vdots & \ddots & \vdots \\
A_{n1}B & \cdots & A_{nn}B\\
\end{bmatrix}.$$
Tensor product satisfies $(A\otimes B)\cdot(C\otimes D)=(AC)\otimes(BD)$. We denote by $A^{\otimes n}$ the tensor product of $n$'s matrices $A$. We use the Dirac notation $|\cdot\rangle$ to denote a column vector in $\mathbb{C}^n$, and its conjugate is $\langle\cdot|=|\cdot\rangle^{\dag}$. $|\phi\rangle|\psi\rangle$ and $\langle\phi|\psi\rangle$ represent $|\phi\rangle \otimes|\psi\rangle$ and $\langle\phi||\psi\rangle$, respectively.
	
As usual, $\mathbb{C}^n$ is an $n$-dimension Hilbert space with usual inner product   $\langle \phi|\psi \rangle$ for $|\phi\rangle, |\psi\rangle\in \mathbb{C}^n$, and
	 its norm  $\||\phi\rangle\|=\sqrt{\langle \phi|\phi \rangle}$. Given a finite vector set $B=\{|q_i\rangle: i=1,\dots, n\}$, we denote by $\mathcal{H}(B)$ the Hilbert space $\{\Sigma_i\alpha_i |q_i\rangle:\text{for all }  \alpha_i\in \mathbb{C}\}$, that is spanned by $B$ (or generated by $B$).

	Quantum states can be described by the vectors in $\mathbb{C}^n$ with norm 1, while evolutions of quantum states can be described by unitary matrices. A projective measurement can be described by a set of projective matrices $\{P_i\}$, where $\Sigma_i P_i=I$ and $P_iP_j=O$ if $i\not=j$. If we use $\{P_i\}$ to measure a quantum state $|\phi\rangle$, the probability of obtaining outcome $i$ is $\|P_i|\phi\rangle\|^2$, and the quantum state collapses to $P_i|\phi\rangle/\|P_i|\phi\rangle\|$ if the outcome is $i$.

\subsection{One-way finite automata}
 	In this subsection, we serve to recall the definitions of one-way (quantum) finite automata including DFA, PFA, MO-1QFA, MM-1QFA, multi-letter 1QFA and 1QFAC. We assume the readers are familiar with the basic notations in formal languages and automata theory \cite{hopcroft1979}. For the sake of convenience, we denote by $e_\epsilon$ the empty string, and denote by $|x|$ the length of string $x$.

\begin{definition}[DFA]
	A DFA $\mathcal{A}$ is defined as a quintuple $\mathcal{A}=(S, \Sigma, s_0, \delta, F)$, where $S$ is a set of finite states, $s_0$ is the initial state, $\Sigma$ is a finite input alphabet, $\delta: S\times \Sigma \rightarrow S$ is a transition map, and $F\subseteq S$ is the set of accepting states.
\end{definition}

	Given a transition map $\delta$, we extend $\delta$ as follows:
$$
\begin{cases}
\delta(s,e_\epsilon)=s,\\
\delta(s,w\sigma)=\delta(\delta(s,w),\sigma), &w\in \Sigma^*,\sigma\in\Sigma.
\end{cases}
$$

	For any input string $x$, if $\delta(s_0,x)\in F$, then $\mathcal{A}$ accepts $x$, otherwise $\mathcal{A}$ rejects $x$.

\begin{definition}[PFA \cite{tzeng1992a}]
	A PFA $\mathcal{A}$ is defined as a quintuple $\mathcal{A}=(S,\Sigma,M,\rho,F)$, where $S$ is a set of finite states, $\rho\in\mathbb{R}^{1\times |S|}$ is the initial-state distribution, $\Sigma$ is a finite input alphabet, $M\in\mathbb{R}^{|S|\times |S|}$ is a stochastic matrix where the $(i,j)$th entry of $M$ represents the probability of the $i$th state in $S$ transiting to the $j$th state in $S$, and $F\subseteq S$ is a set of accepting states.
\end{definition}

	Let $\eta_F$ be a vector in $\mathbb{R}^{|S|}$ such that for $1\leq i\leq |S|$, if $s_i\in F$, then the $i$th entry of $\eta_F$ is 1, otherwise the $i$th entry of $\eta_F$ is 0. Given an input string $x=\sigma_1\cdots \sigma_n \in \Sigma^*$, the probability of $\mathcal{A}$ accepting $x$ is $\rho M(\sigma_1)\dots M(\sigma_n) \eta_F$.

\begin{definition}[MO-1QFA]
	An MO-1QFA $\mathcal{A}$ is defined as a quintuple $\mathcal{A}=(Q, \Sigma, |\psi_0\rangle, \{U_{\sigma}\}_{\sigma \in \Sigma}, Q_{acc})$, where $Q$ is a set of finite states that form an orthonormal basis in the Hilbert space $\mathbb{C}^{|Q|}$ (we call them ``quantum basis states" ), $|\psi_0\rangle$ is the initial quantum state which is a unit vector in $\mathbb{C}^{|Q|}$, $\Sigma$ is a finite input alphabet, $U_{\sigma}$ is a unitary matrix on $\mathbb{C}^{|Q|}$ for each $\sigma \in \Sigma$ and $Q_{acc}\subseteq Q$ is the set of accepting states.
\end{definition}

	For any input string $x=\sigma_1\cdots \sigma_n \in \Sigma^*$, MO-1QFA $\mathcal{A}$ works as follows: at the beginning, $\mathcal{A}$ is in the initial quantum state $|\psi_0\rangle$, and the tape head will scan the input string from left to right. When the character $\sigma_i$ is being scanned, the unitary transformation $U_{\sigma_i}$ acts on the current quantum state and then the current quantum state changes. When the last letter $\sigma_n$ has been scanned, $\mathcal{A}$ is in the final quantum state $|\psi_x \rangle=U_{\sigma_n}\cdots U_{\sigma_1}|\psi_0\rangle$. Finally, a measurement is performed on $|\psi_x \rangle$ and the accepting probability is
$${\rm{Prob}}_{\mathcal{A},acc}(x)=\left\| P_{acc} |\psi_x\rangle \right\|^2,$$
where $P_{acc}=\Sigma_{a\in Q_{acc}}|a\rangle \langle a|$ is the projector onto the subspace generated by $Q_{acc}$.

\begin{definition}[multi-letter 1QFA]
	A $k$-letter 1QFA $\mathcal{A}$ is defined as a quintuple $\mathcal{A}=(Q, \Sigma, |\psi_0\rangle, \{U_{w}\}_{w \in (\{\Lambda\}\cup \Sigma)^k}, Q_{acc})$, where $Q, \Sigma, |\psi_0\rangle$, $Q_{acc}$ are the same as those in MO-1QFA, $U_{w}$ is a unitary operator for each $w \in (\{\Lambda\}\cup \Sigma)^k$, and $\Lambda\not\in\Sigma$ is the blank letter.
\end{definition}

	The working process of a $k$-letter 1QFA is almost the same as MO-1QFA, except that each evolution of its quantum state depends on the last $k$ letters received currently. A $1$-letter 1QFA is actually an MO-1QFA. For any input string $x=\sigma_1\cdots \sigma_n$, unlike MO-1QFA, the final quantum state of $k$-letter 1QFA $\mathcal{A}$ is
\begin{equation}
|\psi_x\rangle=
\begin{cases}\label{final_state_ML-1QFA}
U_{\Lambda^{k-n}\sigma_1\sigma_2\cdots \sigma_n}\cdots U_{\Lambda^{k-2}\sigma_1\sigma_2}U_{\Lambda^{k-1}\sigma_1}|\psi_0\rangle,&\text{if}\ n<k,\\
U_{\sigma_{n-k+1}\sigma_{n-k+2}\cdots \sigma_n}\cdots U_{\Lambda^{k-2}\sigma_1\sigma_2}U_{\Lambda^{k-1}\sigma_1}|\psi_0\rangle,&\text{if}\ n\geq k.
\end{cases}
\end{equation}

\begin{definition}[MM-1QFA]
	An MM-1QFA $\mathcal{A}$ is defined as a 6-tuple $\mathcal{A}=(Q, \Sigma, |\psi_0\rangle, \{U_{\sigma}\}_{\sigma \in \Sigma\cup\{\$\}}, Q_{acc}, Q_{rej})$, where $Q, \Sigma, |\psi_0  \rangle$, $U_{\sigma}$ are the same as those in MO-1QFA, $\$ \not\in \Sigma$ denotes the right end-mark, $Q_{acc}\subseteq Q$ and $Q_{rej}\subseteq Q$ satisfying $Q_{rej}\cap Q_{acc}=\emptyset$ denote accepting and rejecting states, respectively. In addition, $Q_{non}=Q-Q_{acc}-Q_{rej}$ denotes non-halting states.
\end{definition}

 	The working process of an MM-1QFA is similar to MO-1QFA, except that after each unitary transformation, a projective measurement consisting of $\{P_{acc},P_{rej},P_{non}\}$ is performed on the current state, where $P_{acc}=\Sigma_{a\in Q_{acc}}|a\rangle \langle a|$, $P_{rej}=\Sigma_{a\in Q_{rej}}|a\rangle \langle a|$, $P_{non}=\Sigma_{a\in Q_{non}} |a\rangle \langle a|$. If the measurement outcome is ``accept" (or ``reject"), $\mathcal{A}$ halts and accepts (or rejects) the input string, otherwise $\mathcal{A}$ continues working until the last measurement is performed.

\begin{definition}[1QFAC\cite{QIU20151QFAC}]
	A 1QFAC $\mathcal{A}$ is defined as a 9-tuple $\mathcal{A}=(S, Q, \Sigma, \Gamma, s_0, |\psi_0\rangle, \delta, \mathbb{U}, \mathcal{M})$, where $Q,\Sigma, |\psi_0\rangle$ are the same as those in MO-1QFA, $S$ is a finite set of classical states, $\Gamma$ is the finite output alphabet, $s_0$ is the initial classical state, $\delta: S\times \Sigma \rightarrow S$ is the classical transition map, $\mathbb{U}=\{U_{s\sigma}\}_{s\in S,\sigma \in \Sigma}$ where $U_{s\sigma}$ is a unitary operator for each $s\in S$ and $\sigma\in \Sigma$, and $\mathcal{M}=\{\mathcal{M}_s\}_{s\in S}$ where $\mathcal{M}_s=\{P_{s,\gamma}\}_{\gamma\in \Gamma}$ is a projective measurement over $\mathcal{H}(Q)$ with outcomes in $\Gamma$ for each $s\in S$.
\end{definition}

	For any input string $x=\sigma_1\cdots \sigma_n$, $\mathcal{A}$ works as follows: $\mathcal{A}$ starts at the initial classical state $s_0$ and initial quantum state $|\psi_0\rangle$. The tape head will scan the input string from left to right. When the character $\sigma_i$ is being scanned, the transformation $U_{s\sigma_i}$ acts on the current quantum state, where $s$ is the current classical state and $s$ is changed to $\delta(s,\sigma_i)$. When $\sigma_n$ has been scanned, the final quantum state and classical state are $|\psi_x \rangle$ and $s_x$, respectively. Finally, the measurement $\mathcal{M}_{s_x}$ is performed on $|\psi_x \rangle$ and the probability of producing $\gamma$ is
$${\rm{Prob}}_{\mathcal{A},\gamma}(x)=\left\| P_{{s_x},\gamma} |\psi_x \rangle \right\|^2,$$
where $s_x=\delta(s_0,x)$ and $|\psi_x \rangle=U_{\delta(s_0,\sigma_1\cdots \sigma_{n-1})\sigma_n}\cdots U_{\delta(s_0,\sigma_1)\sigma_2}U_{s_0\sigma_1}|\psi_0 \rangle$.

	In this paper, we just consider $\Gamma=\{\text{``accept"},\text{``reject"}\}$. We say that a language $L$ is recognized by a QFA with isolated cut-point if ${\rm{Prob}}_{acc}(x)\geq \lambda+\epsilon$ for $x\in L$ and ${\rm{Prob}}_{acc}(x)\leq \lambda-\epsilon$ for $x\not\in L$ for some $\lambda$ in $(0,1)$ and $\epsilon>0$. More precisely, we say $L$ is recognized by the QFA with cut-point $\lambda$ isolated by $\epsilon$. We say a language $L$ is recognized by a QFA with one-side error $\epsilon$ if ${\rm{Prob}}_{acc}(x)=1$ for $x\in L$ and ${\rm{Prob}}_{acc}(x)\leq\epsilon$ for $x\not\in L$.

\section{State complexity of 1QFAC}\label{sec:state-complexity}
	State complexity is an important subject of finite automata, which  reflects the computational and space complexity and how much computing resources is required to some extent. In this section, we further study state complexity problems concerning 1QFAC. We give the quantitative relationships between the state number of 1QFAC and DFA, clarify the trade-offs between quantum basis states and classical states, and give a lower bound of the classical state number of 1QFAC for recognizing any given regular language.


\subsection{Quantitative relationships between the state number of 1QFAC and DFA}
	Given a regular language $L$, if a 1QFAC recognizes $L$ with isolated cut-point and with $n$ quantum basis state and $k$ classical states, and the minimal DFA of $L$ has $m$ states, then what are the quantitative relationships between $n$,$k$ and $m$? Qiu et al. \cite{QIU20151QFAC}  proved that $m=2^{O(kn)}$. Here our purpose is to exponentially optimize the bound to $m=k2^{O(n)}$. In Section \ref{sec: state_succinctness} we further show that our bound is actually tight.

\begin{definition}
Given a 1QFAC $\mathcal{A} = (S, Q, \Sigma,\Gamma, s_0, |\psi_0 \rangle, \delta , \mathbb{U}, \mathcal{M} )$ and a string $x=x_1x_2\cdots x_n$, where $\mathbb{U}=\{U_{s\sigma}\}_{s\in S,\sigma \in \Sigma}$, we introduce the following notations for 1QFAC $\mathcal{A}$.
\begin{itemize}
\item $s_y$: $\forall y\in \Sigma^*,s_y=\delta(s_0,y)$. The notation can be used in DFA as well.
\item $U_{s,x}$: $\forall s\in S,U_{s,x}=U_{s_{n-1}x_n}\cdots U_{s_1x_2}U_{sx_1}$, where $s_1=\delta(s,x_1)$, $s_i=\delta(s_{i-1},x_i),i=2,3,\dots,n$, we specify $U_{s,e_\epsilon}=I$.
\item $|\psi_x \rangle$: $|\psi_x \rangle=U_{s_0,x}|\psi_0 \rangle$.
\item $|\phi_x  \rangle$: $|\phi_x  \rangle=|s_x \rangle|\psi_x \rangle$, where $\{|s\rangle:s\in S\}$ is an orthonormal base of $\mathbb{C}^{|S|}$. 
\end{itemize}
\end{definition}



	First we need the following lemmas.

\begin{lemma}[\cite{QIU20151QFAC}]\label{number_of_elements}
	Let $V_{\theta}\subseteq \mathbb{C}^n$ satisfy that for any $|\phi_1\rangle, |\phi_2\rangle \in V_{\theta},\left\| |\phi_1 \rangle-|\phi_2\rangle \right\|\geq \theta$ holds if $|\phi_1\rangle \not=|\phi_2\rangle$. Then $|V_{\theta}|\leq (1+\frac{2}{\theta})^{2n}.$
\end{lemma}

\begin{lemma}\label{vector_distance}
Let P be a projective operator on $\mathbb{C}^n $, and let $|\phi_1 \rangle,|\phi_2 \rangle$ be unit vectors in $\mathbb{C}^n$. Then $| \left\|P|\phi_1\rangle \right\|^2-\left\|P|\phi_2\rangle \right\|^2| \leq  2\left\|\ |\phi_1 \rangle-|\phi_2\rangle\ \right\|$.
\end{lemma}
\begin{proof}
$\left| \left\|P|\phi_1\rangle \right\|^2-\left\|P|\phi_2\rangle \right\|^2  \right|
=\left| \left\|P|\phi_1\rangle \right\|-\left\|P|\phi_2\rangle \right\|      \right| \cdot
 \left| \left\|P|\phi_1\rangle \right\|+\left\|P|\phi_2\rangle \right\|      \right|
\leq \left\|P|\phi_1\rangle -P|\phi_2\rangle \right\| \cdot 2
\leq$\\  $2\left\|\ |\phi_1 \rangle-|\phi_2\rangle\ \right\|$.
\end{proof}

	We recall Myhill-Nerode theorem \cite{hopcroft1979} that characterizes the relationship between the state number of  minimal DFA and the number of equivalence classes derived from its corresponding regular language.

\begin{definition}[\cite{QIU20151QFAC}]
Let $L$ be a regular language. Then an equivalence relation ``$\equiv_L$" is defined as: $\forall x,y\in \Sigma^*$, if $\forall z\in \Sigma^*,xz\in L$ if and only if $yz\in L$, then $x\equiv_L y$, otherwise $x\not\equiv_L y$.
\end{definition}

	 By the equivalence relation ``$\equiv_L$", $\Sigma^*$ is partitioned into some equivalence classes.

\begin{lemma}[Myhill-Nerode theorem\cite{hopcroft1979}] \label{Myhill-Nerode theorem}
	Given any regular language $L\subseteq \Sigma^*$, then the number of equivalence classes of $\Sigma^*$ induced by the equivalence relation ``$\equiv_L$" equals to the number of states in the minimal DFA accepting $L$.
\end{lemma}

	Now we are ready to improve the bound between 1QFAC and DFA mentioned above.

\begin{lemma}\label{number_of_class}
	Let $L$ be a regular language. 1QFAC $\mathcal{A}=(S, Q, \Sigma,\Gamma, s_0, |\psi_0 \rangle, \delta , \mathbb{U}, \mathcal{M} )$ recognizes $L$ with cut-point $\lambda$ isolated by $\epsilon$. Let $A_s=\{w: w\in\Sigma^*\ and\ s_w=s\}$ where $s\in S$, and let $t_s$ denote the number of equivalence classes of partitioning $A_s$ by the equivalence relation ``$\equiv_L$". Then $t_s\leq (1+\frac{2}{\epsilon})^{2n} $, where $n=|Q|$.
\end{lemma}
\begin{proof}
	For any $s\in S$, $\forall x,y \in A_s$, if $x\not\equiv_L y$, then $\exists z\in \Sigma^*$ such that $xz\in L$ whereas $yz\not\in L$ or $xz\not\in L$ whereas $yz\in L$. Without loss of generality, we assume the former case holds. That is,
\begin{equation}
{\rm Prob}_{\mathcal{A},acc}(xz)=\left\|P_{\delta(s,z),acc}U_{s,z}|\psi_x\rangle  \right\|^2\geq \lambda+\epsilon \text{,}
\end{equation}
\begin{equation}
{\rm Prob}_{\mathcal{A},acc}(yz)=\|P_{\delta(s,z),acc}U_{s,z}|\psi_y\rangle  \|^2\leq \lambda-\epsilon \text{.}
\end{equation}
By  \totalref{Lemma}{vector_distance}, we have
\begin{equation}\label{distance_gap}
\|\ |\psi_x \rangle-|\psi_y\rangle\ \|=\|\ U_{s,z}|\psi_x \rangle-U_{s,z}|\psi_y\rangle\ \|\geq \frac{1}{2}[(\lambda+\epsilon)-(\lambda-\epsilon)]=\epsilon .
\end{equation}
By  \totalref{Lemma}{number_of_elements}, we obtain
\begin{equation}
t_s\leq(1+\frac{2}{\epsilon})^{2n}.
\end{equation}
\end{proof}

 Now we have the following theorem that gives a better relationship between the state number of minimal DFA and 1QFAC for recognizing any regular language.
\begin{theorem}\label{1QFAC_size}
	Let $L$ be a regular language, and the minimal DFA accepting $L$ has $m$ states. Given cut-point $\lambda$ and isolation $\epsilon$, if a 1QFAC with $k$ classical states and $n$ quantum basis states recognizes $L$ with cut-point $\lambda$ isolated by $\epsilon$, then $m=k2^{O(n)}$ .
\end{theorem}
\begin{proof}
	Let the 1QFAC be $\mathcal{A}=(S, Q, \Sigma,\Gamma, s_0, |\psi_0 \rangle, \delta , \mathbb{U}, \mathcal{M} )$. By  \totalref{Lemma}{Myhill-Nerode theorem} and Lemma \ref{number_of_class}, we have $m\leq \sum_{s\in S}t_s\leq k(1+\frac{2}{\epsilon})^{2n}$, where $t_s$ is defined in  \totalref{Lemma}{number_of_class}. That is $m=k2^{O(n)}$.
\end{proof}

	As a comparison, the bound shown in \cite{QIU20151QFAC} is $m=2^{O(kn)}$. Let $k,n,m$ be the same as those in  \totalref{Theorem}{1QFAC_size}. Then we have the following corollaries.

\begin{corollary}
	$\max(k,n)=\Omega(\log m)$.
\end{corollary}
\begin{corollary}
	If $k=O(\sqrt{m})$, then $n=\Omega(\log m)$.
\end{corollary}

	In other words, if a 1QFAC has fewer classical states than the minimal DFA for recognizing some regular language $L$ in some degree ($k=O(\sqrt{m})$), then it needs at least $n=\Omega(\log m)$ quantum basis states.

\subsection{Upper bound on reducing quantum basis states by adding classical states}
	1QFAC is a model that integrates quantum states and classical states. Naturally, it is attractive to study the   trade-offs between quantum basis states and classical states. We focus on two cases: (1) How many classical states are needed if its quantum basis states of a 1QFAC are reduced without changing its recognition ability? (2) How many quantum basis states are needed if its classical states of a 1QFAC are cut down?

	In this section, we give an upper bound showing that how many classical states are needed if the quantum basis states of a 1QFAC are reduced without changing its recognition ability. In particular, we prove the bound is tight, that is, it is attainable. However, the proof of the bound being tight is complicated, so we need to define some concepts and give a number of results in advance.

First we define the following concepts.

\begin{itemize}
\item \emph{Unary DFA} $\mathcal{A}=(S, \Sigma, s_0, \delta, F)$: the alphabet of $\mathcal{A}$ contains only one element. We define its alphabet as $\{0\}$ usually, and the following notations are related to unary DFA $\mathcal{A}=(S, \Sigma, s_0, \delta, F)$.
\item \emph{s cycle}: $s$ is a state of $\mathcal{A}$ and $\exists x\in \Sigma^*$ such that $s_x=s$ (i.e. $s$ is reachable). it refers to the process of $\mathcal{A}$ reading several 0s from $s$ back to $s$.
\item \emph{Minimal s cycle}: the process of $\mathcal{A}$ reading several 0s from $s$ back to $s$ for the first time.
\item \emph{The length of an s cycle}: the number of 0s read by $\mathcal{A}$ during the $s$ cycle. Similarly, we can define ``\emph{the length of the minimal $s$ cycle}".
\item \emph{The length of the minimal cycle}: it is clear that the length of the minimal $s$ cycle are the same for any legal $s$, so we call them to ``\emph{the length of the minimal cycle}".
\end{itemize}

\begin{remark}
	Though the classical part of a unary 1QFAC is not a unary DFA, we can still use the above notations for the classical part of the 1QFAC if we only consider the classical states and their transitions.
\end{remark}

Now we present  the following theorem.

\begin{theorem}\label{reduce_quantum}
Let $L$ be a regular language that is recognized by a 1QFAC with cut-point $\lambda$ isolated by $\epsilon$ and with k classical states and n quantum basis states. Let $k'$ be the minimal number of classical states among all 1QFAC which recognizes $L$ with cut-point $\lambda$ isolated by $\epsilon$ and with fewer than $n$ quantum basis states. Then $k'=k2^{O(n)}$. In addition, the bound is tight.
\end{theorem}

For giving  the proof (especially the  the bound being tight), we need a number of lemmas and propositions.

\begin{lemma}\label{unary_DFA}
	Given a unary DFA $\mathcal{A}=(S, \Sigma, s_0, \delta, F)$ which recognizes $L$, if it reads 0s infinitely, then it holds that
\begin{enumerate}[(I)]
\item The length of an $s$ cycle of $\mathcal{A}$ is divisible by the length of the minimal cycle of $\mathcal A$, where $s\in S$.
\item If $L=L(d)=\{0^{zd}:z=0,1,\dots\}$ where $d\in \mathbb{Z}^+$ is a constant, then the length of the minimal cycle of $\mathcal A$ is divisible by $d$.
\end{enumerate}
\end{lemma}
\begin{proof}
	(I) is evident. We prove (II). Consider the minimal $s_w$ cycle of $\mathcal A$, where $w\in L$. Suppose the length of the minimal $s_w$ cycle is $n$, then we have $w0^n\in L$. Since $d||w|$ and $d||w0^n|$, we obtain $d|n$, where $d|n$ means that $n$ is divisible by $d$.
\end{proof}

	Let $L$ be a regular language and 1QFAC $\mathcal{A}=(S, Q, \Sigma,\Gamma, s_0, |\psi_0 \rangle, \delta , \mathbb{U}, \mathcal{M} )$. We define ``$\equiv_{L,s}$" as: $\forall x,y\in \Sigma^*$, $x\equiv_{L, s}$ iff $x\equiv_L y$ and $s_x=s_y$.

	$\Sigma^*$ can be partitioned into several equivalence classes by ``$\equiv_{L, s}$", and we denote by $[x]$ the equivalence class where the string $x$ is in. We regard these equivalence classes as states in DFA $\mathcal{D}=(S',\Sigma, s_0', \delta', F)$, and define $s_0'=[e_\epsilon]$, $\delta'([x],\sigma)=[x\sigma]$, $F=\{[w]: w\in L\}$. We call DFA $\mathcal{D}$ is the derived DFA of $\mathcal{A}$ and $L$.

\begin{lemma}\label{cycle_length}
	Let unary 1QFAC $\mathcal{A}=(S, Q, \Sigma,\Gamma, s_0, |\psi_0 \rangle, \delta , \mathbb{U}, \mathcal{M} )$ with $n$ quantum basis states recognize $L$ with cut-point $\lambda$ isolated by $\epsilon$. Let DFA $\mathcal{D}$ be the derived DFA of $\mathcal{A}$ and $L$, let $l_{\mathcal{D}}$ and $l_{\mathcal{A}}$ be the length of the minimal cycle of $\mathcal{D}$ and the classical part of $\mathcal{A}$, respectively. Then $l_{\mathcal{D}}=l_{\mathcal{A}}l_0$, where $l_0\in \mathbb{Z}^+$ and $l_0\leq(1+\frac{2}{\epsilon})^{2n}$.
\end{lemma}
\begin{proof}
	Let $\mathcal{A}=(S, Q, \Sigma,\Gamma, s_0, |\psi_0 \rangle, \delta , \mathbb{U}, \mathcal{M} )$. Suppose $\mathcal{A}$ and $\mathcal{D}$ read 0s simultaneously. According to the definition of ``$\equiv_{L,s}$", if $\mathcal{D}$ finishes a minimal $[x]$ cycle, then the classical part of $\mathcal{A}$ finishes an $s_x$ cycle as well. Hence, by (I) of \totalref{Lemma}{unary_DFA}, we have $l_0\in \mathbb{Z}^+$, and with \totalref{Lemma}{number_of_class}, we obtain $l_0\leq(1+\frac{2}{\epsilon})^{2n}$.
\end{proof}

	The following lemma gives the operation properties of 1QFAC and MO-1QFA.

\begin{lemma}[\cite{QIU20151QFAC}]\label{1QFAC_construction}
	Let $L_1$ and $L_2$ be regular languages over a finite alphabet $\Sigma$. $L_1$ is recognized by a minimal DFA with $k$ states and $L_2$ is recognized by an MO-1QFA with cut-point $\lambda$ isolated by $\epsilon$ and with $n$ quantum basis states.  Then $L_1\cap L_2, L_1\cup L_2, L_1\backslash L_2$ and $L_2\backslash L_1 $ can be recognized by some 1QFAC with cut-point $\lambda$ isolated by $\epsilon$ and with $k$ classical states and $n$ quantum basis states.
\end{lemma}

	Consider language $L(pq)=L(p)\cap L(q)=\{0^{zp}: z=0,1,2,\dots \}\cap \{0^{zq}: z=0,1,2,\dots \}=\{0^{zpq}: z=0,1,2,\dots\}$, where $p,q$ are prime numbers. The minimal DFA recognizing $L(p)$ has $p$ states. Given cut-point $\lambda$ and isolation radius $\epsilon$ satisfying $\lambda-\epsilon>0$, then there exists an MO-1QFA recognizing $L(q)$ with cut-point $\lambda$ isolated by $\epsilon$ and with $n$ quantum basis states, where $n$ is the minimal number of quantum basis states among all MO-1QFA which recognizes $L(q)$ with cut-point $\lambda$ isolated by $\epsilon$. Actually $n=\Theta(\log q)$, since $n=\Omega(\log q)$ and there exists an MO-1QFA with $\Theta(\log q)$ quantum basis states recognizing $L(q)$ with cut-point $\lambda$ isolated by $\epsilon$ \cite{strengths}. In addition, we restrict that $p>(1+\frac{2}{\epsilon})^{2n}$. By \totalref{Lemma}{1QFAC_construction}, we know that there exists a 1QFAC with $p$ classical states and $n$ quantum basis states recognizing $L(pq)$ with cut-point $\lambda$ isolated by $\epsilon$. The following proposition gives how many classical states are needed for a 1QFAC recognizing $L(pq)$ with fewer than $n$ quantum basis states and with cut-point $\lambda$ isolated by $\epsilon$.

\begin{proposition}\label{amount_of_classical_states}
	$p,q,n,\lambda,\epsilon ,  L(pq)$ are given above and let 1QFAC $\mathcal{A}$ recognize $L(pq)$ with cut-point $\lambda$ isolated by $\epsilon$ and with $k'$ classical states and $n'$ quantum basis states, where $n'<n$. Then $k'\geq pq=p2^{\Theta(n)}$.
\end{proposition}
\begin{proof}
	Since $n=\Theta(\log q)$, we have $q=2^{\Theta(n)}$. Let $\mathcal{A}=(S, Q, \Sigma,\Gamma, s_0, |\psi_0 \rangle, \delta , \mathbb{U}, \mathcal{M} )$, where $\Sigma=\{0\}, |Q|=n', \mathcal{M}=\{\mathcal{M}_s\}_{s\in S}, \mathcal{M}_s=\{P_{s,\gamma}\}_{\gamma\in \Gamma},\mathbb{U}=\{U_{s\sigma}\}_{s\in S,\sigma \in \Sigma}$.  We prove it by contradiction. Assume $k'<pq$.
	
	Let DFA $\mathcal{D}$ be the derived DFA of $\mathcal{A}$ and $L$, and let $l_{\mathcal{D}}$ and $l_{\mathcal{A}}$ be the length of the minimal cycle of $\mathcal{D}$ and the classical part of $\mathcal{A}$, respectively. By \totalref{Lemma}{cycle_length}, we have $l_{\mathcal{D}}=l_{\mathcal{A}}l_0$, where $l_0 \in \mathbb{Z}^+$ and $l_0\leq (1+\frac{2}{\epsilon})^{2n'}<p$. By (II) of \totalref{Lemma }{unary_DFA}, we have $pq|l_{\mathcal{D}}$, that is, $pq|l_{\mathcal{A}}l_0$. Considering $l_{\mathcal{A}}\leq k'<pq$ and $l_0<p$, we get $p|l_{\mathcal{A}}, q\nmid l_{\mathcal{A}}$ and $q|l_0$.

	There exists $w\in L$ such that $\delta(s_w,0^{l_{\mathcal{A}}})=s_w$, since $l_{\mathcal{A}}$ is the length of the minimal cycle of the classical part of $\mathcal{A}$. Let $x_1=w0^{zl_{\mathcal{A}}}$. Then it holds that $\forall z\in \mathbb{Z}^+,|\psi_{x_1}\rangle=(U_{s_w,0^{l_A}})^z|\psi_{w}\rangle$ and $q|z\Rightarrow pq|zl_A \Rightarrow {\rm{Prob}}_{\mathcal{A},acc}(x_1)\geq \lambda+\epsilon,q \nmid z\Rightarrow pq \nmid zl_A \Rightarrow {\rm{Prob}}_{\mathcal{A},acc}(x_1)\leq \lambda-\epsilon$, since $\mathcal{A}$ recognizes $L(pq)$ with cut-point $\lambda$ isolated by $\epsilon$.

	Hence, we can construct an MO-1QFA $\mathcal{A}'$ with $n'$ quantum basis states recognizing $L(q)$ with cut-point $\lambda$ isolated by $\epsilon$: $\mathcal{A}'=(Q', \Sigma, |\psi_0'\rangle, \{U'_{\sigma}\}_{\sigma \in \Sigma}, Q_{acc})$, where $\mathcal{H}(Q')=\mathcal{H}(Q), |\psi_0'\rangle=|\psi_w\rangle, U'_0=U_{s_w,0^{l_A}}, \sum_{a\in Q_{acc}}|a\rangle \langle a|=P_{s_w,acc}$. It can be easily seen that $q|z \Rightarrow {\rm{Prob}}_{\mathcal{A}',acc}(0^z)\geq \lambda+\epsilon$ and $q \nmid z \Rightarrow {\rm{Prob}}_{\mathcal{A}',acc}(0^z)\leq \lambda-\epsilon$.

	It leads to a contradiction since $n$ is the minimal number of quantum basis states among all MO-1QFA which recognize $L(q)$ with cut-point $\lambda$ isolated by $\epsilon$ and $n'<n$. Thus, $k'\geq pq=p2^{\Theta(n)}$.
\end{proof}

 Now we are ready for presenting the proof of Theorem \ref{reduce_quantum}.

 {\bf  The proof of Theorem \ref{reduce_quantum}}:

\begin{proof}
	By  \totalref{Theorem}{1QFAC_size}, we know that there exists a 1QFAC with $k2^{O(n)}$ classical states and $1$ quantum basis states recognizing $L$ with cut-point $\lambda$ isolated by $\epsilon$. So we get $k'=k2^{O(n)}$. In the worst case, $k'=k2^{\Theta(n)}$. Proposition 1 shows that the worst case is possible for some languages. Therefore the bound in \totalref{Theorem}{reduce_quantum} is tight.



\end{proof}

	On the other hand, if we construct a 1QFAC recognizing any given regular language with fewer classical states, how many quantum basis states do we need? It is a pity that sometimes we can not reduce any classical state of some 1QFAC, otherwise we can reduce the classical state number of any 1QFAC to $1$ for any regular language. However, a 1QFAC with only one classical state can be regarded as an MO-1QFA, but  MO-1QFA can not recognize all regular languages with isolated cut-point. Thus, we study the lower bound of the classical states of 1QFAC in next section.

\subsection{Lower bound of the classical states of 1QFAC}
	Studying the lower bound of the classical states of 1QFAC allows us to discover that how many classical states of 1QFAC are needed. We can compare with DFA. In this section, we give several results to determine or estimate the lower bound of the classical state number of 1QFAC.
First we give the following lemmas.

\begin{lemma}[\cite{characterizations}]\label{approximation_of_U}
	Let $U$ be a unitary matrix. Then $\forall \epsilon>0,\forall N\in \mathbb{Z}^+, \exists n\in \mathbb{Z}^+$ and $n>N$ such that $\|(I-U^n)x \|<\epsilon$ holds for any vector $x$ with $\|x\|\leq 1$.
\end{lemma}

\begin{lemma}\label{x_equiv_y}
	Let 1QFAC $\mathcal{A}=(S, Q, \Sigma,\Gamma, s_0, |\psi_0 \rangle, \delta , \mathbb{U}, \mathcal{M} )$ recognize language $L$ with cut-point $\lambda$ isolated by $\epsilon$. $\forall x,y\in \Sigma^*$, if $s_x=s_y$ and $\| |\psi_x\rangle-|\psi_y\rangle \|<\epsilon$, then $x\equiv_{L} y$.
\end{lemma}
\begin{proof}
	Assume that $x\not\equiv_L y$, by inequality (\ref{distance_gap}), we have $\| |\psi_x\rangle-|\psi_y\rangle \|\geq\epsilon$, which contradicts $\| |\psi_x\rangle-|\psi_y\rangle \|<\epsilon$. Thus, $x\equiv_{L} y$.
\end{proof}

	Given string $x=x_1x_2\dots x_n$,  $\forall i,j\in \mathbb{N}$ with $1\leq i\leq j\leq n$, define $x_{[i,j]}=x_{i}\dots x_j$. In general, we specify $x_0=e_\epsilon$. We give a lower bound of the classical state number of 1QFAC which recognizes given finite regular language.

\begin{theorem}\label{finite_RL}
	Given a finite regular language $L$. Let 1QFAC $\mathcal{A}$ recognize $L$ with cut-point $\lambda$ isolated by $\epsilon$ and with $k$ classical states. Then $k\geq\max \limits_{x\in L}|x|+2$.
\end{theorem}
\begin{proof}
	Let $\mathcal{A}=(S, Q, \Sigma,\Gamma, s_0, |\psi_0 \rangle, \delta , \mathbb{U}, \mathcal{M} )$, where $\mathbb{U}=\{U_{s\sigma}\}_{s\in S,\sigma \in \Sigma}$. Let $x$ be the longest string in $L$. By taking $\sigma \in \Sigma$, we have $x\sigma\not \in L$. Suppose the classical states of $\mathcal{A}$ reading $x\sigma$ are  $s_0, s_1, s_2, \cdots,  s_n, s_{n+1}$, in sequence.

	Assume that $k\geq|x|+2$ does not hold, that is, $k<|x|+2$. Then $\exists i,j\in\{0,1,\dots,n,n+1\},i<j$ such that $s_i=s_j$. Let $x'\in \Sigma^*$ satisfy $x=x_{[0,i]}x'$. We consider the following two cases.\\
\textbf{Case 1.} $j\leq n$. This case implies that $n\geq 1$. Take  $y=x_{[0,i]}x_{[i+1,j]}^m x'$ where $m\in\mathbb{Z}^+$. Since $s_i=s_j$, we have $s_y=s_x$ and $|\psi_y\rangle=U_{s_j,x'}U_{s_i,x_{[i+1,j]}}^m|\psi_{x_{[0,i]}}\rangle$. By \totalref{Lemma}{approximation_of_U}, there exists $m>1$ such that $\left\| |\psi_y\rangle-|\psi_x\rangle \right\|=\left\| U_{s_j,x'}((U_{s_i,x_{[i+1,j]}})^{m}-I)|\psi_{x_{[0,i]}}\rangle \right\|<\epsilon$. By \totalref{Lemma}{x_equiv_y}, we obtain $x\equiv_{L} y$. It contradicts that $y\not\in L$ and $x\in L$.\\
\textbf{Case 2.} $j=n+1$. Since $s_i=s_{n+1}$, we have $\delta(s_i,x')=s_n, \delta(s_n,\sigma)=s_{n+1}=s_i$, which implies $\delta(s_n,\sigma x')=s_n$. Take $y=x(\sigma x')^m$. We have $s_y=s_x=s_n$. By \totalref{Lemma}{approximation_of_U}, there exists $m>1$ such that $\left\| |\psi_y\rangle-|\psi_x\rangle \right\|=\| ((U_{s_n,\sigma x'})^m-I)|\psi_x\rangle \|<\epsilon$. By \totalref{Lemma}{x_equiv_y}, we obtain $x\equiv_{L} y$. It also contradicts that $y\not\in L$ and $x\in L$.

	Therefore, $k\geq|x|+2$ holds.
\end{proof}

With Theorem \ref{finite_RL} we have the following corollary describing the relationship between the classical state number of minimal DFA and 1QFAC for any  finite regular language.
\begin{corollary}
	Let $L$ be a finite regular language over alphabet $\Sigma$. Suppose the minimal DFA recognizing $L$ has $m$ states, and $k$ is the minimal number of classical states among all 1QFAC which recognize $L$ with isolated cut-point. Then if $|\Sigma|=1$ or $|L|=1$, we have $k=m$. If $|\Sigma|>1$, we have $2|\Sigma|^{k-2}\geq m$.
\end{corollary}
\begin{proof}
	The former case is clear, so we prove the later case. Let $l$ be the length of the longest string in $L$. Then $m\leq 1+\sum_{i=0}^l |\Sigma|^i=|\Sigma|^l+\sum_{i=0}^{l-1} |\Sigma|^i+1=|\Sigma|^l+\frac{|\Sigma|^l-1}{|\Sigma|-1}+1\leq 2|\Sigma|^l\leq 2|\Sigma|^{k-2}$.
\end{proof}

	 Next we prove the bound in \totalref{Theorem}{finite_RL} is exact lower bound.

\begin{theorem}\label{exactly_recognize}
	Given a finite regular language $L$. Let $l=\max \limits_{x\in L}|x|$. Then there exists a 1QFAC recognizing $L$ exactly with $l+2$ classical states and $m^l$ quantum basis states, where $m=|\Sigma|$.
\end{theorem}
\begin{proof}
	Let the 1QFAC be $\mathcal{A}=(S, Q, \Sigma,\Gamma, s_0, |\psi_0 \rangle, \delta , \mathbb{U}, \mathcal{M} )$, where $\Sigma=\{0,\dots,m-1\}, \mathcal{M}=\{\mathcal{M}_s\}_{s\in S}, \mathcal{M}_s=\{P_{s,\gamma}\}_{\gamma\in \Gamma},\mathbb{U}=\{U_{s\sigma}\}_{s\in S,\sigma \in \Sigma}$. For the sake of convenience, we denote by $|0\rangle_m$ and $I_m$ the zero vector and the identity matrix in $m$-dimension Hibert space, respectively. We define $\mathcal{A}$ as follows:
\begin{itemize}
\item $S=\{s_0,s_1,\dots,s_{l+1}\}$.
\item $\forall \sigma\in \Sigma$, $\delta(s_i,\sigma)=
\begin{cases}
s_{i+1},& i \in \{0,1,\dots,l\},\\
s_{l+1},& i=l+1.
\end{cases}$
\item $|\psi_0 \rangle=|0 \rangle_m^{\otimes l}$.
\item $P_{s_{l+1},acc}=0I_m^{\otimes l}, P_{s_{i},acc}=\sum_{w\in L\wedge |w|=i}|w\rangle \langle w|\otimes I_m^{\otimes (l-i)}, i=1,\dots,l$, where $|w\rangle=|w_1\rangle\cdots|w_i\rangle$ if $w=w_1\cdots w_i$ and $w_j\in \Sigma,j=1,\dots,i$. $P_{s_0,acc}=
\begin{cases}
I_m^{\otimes l},& e_\epsilon \in L,\\
0I_m^{\otimes l},& e_\epsilon \not\in L.
\end{cases}$
\item $\forall \sigma\in \Sigma$ and $i \in \{0,1,\dots,l-1\}$, $U_{s_{i}\sigma}|0\rangle_m=|\sigma\rangle_m$.
\end{itemize}

	The 1QFAC $\mathcal{A}$ works as follows: For input string $w\in \Sigma^*$, if $|w|\geq l+1$, then the final classical state after $\mathcal{A}$ reading $w$ is $s_{l+1}$ and $w$ will be rejected with probability 1. Otherwise the final classical state is $s_{|w|}$ and the final quantum state before measurement is $|w\rangle|0\rangle^{\otimes l-|w|}$. If $w\in L$, then $P_{s_{|w|},acc}|w\rangle|0\rangle^{\otimes l-|w|}=|w\rangle|0\rangle^{\otimes l-|w|}$ and ${\rm{Prob}}_{\mathcal{A},acc}(w)=1$. If $w\not\in L$, then ${\rm{Prob}}_{\mathcal{A},acc}(w)=0$. Therefore, $\mathcal{A}$ recognizes $L$ exactly.
\end{proof}

	We have given an exact lower bound of the classical state number of 1QFAC for recognizing any given finite regular language. Similarily, we can generalize \totalref{Theorem}{finite_RL} and give a more general result.

	Given a regular language $L$, for any $x\in \Sigma^*$, we say $x$ is $L$ \emph{self-reachable} if there exists $y\in \Sigma^+$ such that $x\equiv_L xy$, otherwise we say $x$ is not $L$ self-reachable.

\begin{lemma}\label{sx!=sy}
	Let 1QFAC $\mathcal{A}=(S, Q, \Sigma,\Gamma, s_0, |\psi_0 \rangle, \delta , \mathbb{U}, \mathcal{M} )$ recognize language $L$ with cut-point $\lambda$ isolated by $\epsilon$. Let $x,y,z\in\Sigma^*$ satisfy that $z$ is not $L$ self-reachable, and there exists $x',z'\in\Sigma^*$, $|z'x'|>0$ such that $z=xx',y=zz'$. Then $s_x\not=s_y$.
\end{lemma}
\begin{proof}
	Assume that $s_x=s_y$, we have $\delta(s_z,z'x')=\delta(\delta(s_z,z'),x')=\delta(s_y,x')=\delta(s_x,x')=s_z$. Similar to the proof of case 2 in \totalref{Theorem}{finite_RL}, there exists $m\in \mathbb{Z}^+$ such that $z\equiv_L z(z'x')^m$. It contradicts that $z$ is not $L$ self-reachable. Hence, $s_x\not=s_y$.
\end{proof}

	Given a regular language $L$. Let string $x=x_1x_2\dots x_n$ and specify $x_0=e_\epsilon$. Then there exists a minimal positive integer $m$ such that $0= i_0<i_1<\cdots<i_m=n+1$ and given any $k\in \{0,1,2,\dots,m-1\}$, all elements in $\{x_{[0,j]}|j\in \mathbb{N},i_k\leq j< i_{k+1}\}$ are $L$ self-reachable or none of them are $L$ self-reachable. Based on the above division, we get $m$ sets. Let $A_1,A_2,\cdots,A_r$ be the sets whose elements are not $L$ self-reachable in these $m$ sets and the elements in the other $m-r$ sets are $L$ self-reachable. We define $C_L(x)=\sum_{j=1}^{r}|A_j|+m-r$.

\begin{theorem}\label{infinite_RL}
	Given a regular language $L$, if a 1QFAC $\mathcal{A}$ recognizes $L$ with isolated cut-point and with $k$ classical states, then for any $x\in \Sigma^+, k\geq C_L(x)$.
\end{theorem}
\begin{proof}
	Let $\mathcal{A}=(S, Q, \Sigma,\Gamma, s_0, |\psi_0 \rangle, \delta , \mathbb{U}, \mathcal{M} )$, $x=x_1x_2\dots x_n$ and specify $x_0=e_\epsilon$. Let $A_1,A_2,\cdots,A_r$ be defined above, we regard each element in $\bigcup\limits_{i=1}^rA_i$ as a set individually. Together with the $m-r$ sets whose elements are $L$ self-reachable, we get $C_L(x)$ sets.

	Assume that $k<C_L(x)$, by the pigeonhole principle, we know there exists $i,j\in\{0,1,\cdots,n\},i<j$ such that $x_{[0,i]},x_{[0,j]}$ belong to two distinct sets of the above $C_L(x)$ sets, respectively, and satisfy $s_{x_{[0,i]}}=s_{x_{[0,j]}}$. Based on the construction of these $C_L(x)$ sets, there exists $t\in \mathbb{N}$ such that $i\leq t\leq j$ and $x_{[0,t]}$ is not $L$ self-reachable, which contradicts \totalref{Lemma}{sx!=sy}. Therefore, $k\geq C_L(x)$.
\end{proof}

	Theorem \ref{infinite_RL} gives a method to estimate the lower bound of the classical state number of a 1QFAC for recognizing any given regular language. When the given regular language $L$ is finite, we have $\max \limits_{x\in L}C_L(x)=\max \limits_{x\in L}|x|+2$, which is the same as the bound in Theorem \ref{finite_RL}.

\section{State succinctness of 1QFAC}\label{sec:state_succinctness}
\label{sec: state_succinctness}
	In this section, we show that 1QFAC are exponentially more succinct than DFA and PFA for recognizing some regular languages that can not be recognized by MO-1QFA, MM-1QFA or multi-letter 1QFA.

	Consider the languages $L(h,p)=(1^*00^*10^h)^*\cap \{w:|w|=kp,k=0,1,2,\dots\}$ over \{0,1\}, where $h\in \mathbb{Z}^+$ and $p$ is a prime number. We first give the minimal DFA recognizing $L(h,p)$ and prove $L(h,p)$ can not be recognized by any  MO-1QFA, MM-1QFA and multi-letter 1QFA. Then, we prove that any PFA recognizing $L(h,p)$ has at least $p$ states. Finally, we show that there exists a 1QFAC recognizing $L(h,p)$ with $h+3$ classical states and $\Theta(\log p)$ quantum basis states.

\subsection{The minimal DFA recognizes $L(h,p)$}
	We first give DFA $\mathcal{D}_1=(Q, \Sigma, \delta, q_0, F)$ that recognizes $(1^*00^*10^h)^*$, where $\Sigma=\{0,1\}$, shown in Figure \ref{fig:DFA1} (the all-rejecting state $q_r$ is omitted).

\begin{figure}[h]
	\centering
	\includegraphics[width=0.7\textwidth]{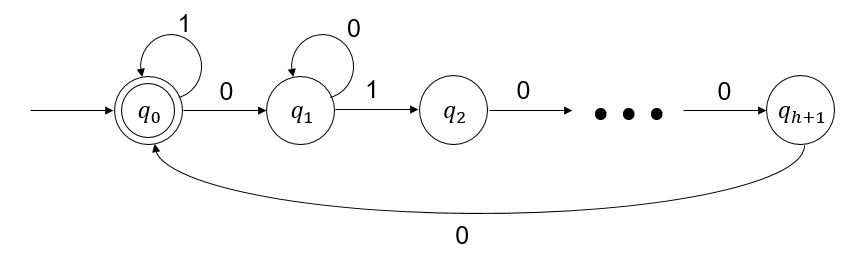}
	\caption{\label{fig:DFA1}The minimal DFA that recognizes $(1^*00^*10^h)^*$}
\end{figure}

	Then we give DFA $\mathcal{D}_2=(Q', \Sigma, \delta', q_{0,0}, F')$ that recognizes $L(h,p)$, where
\begin{itemize}
\item $Q'=\{q_r\}\bigcup\{q_{i,j}: i=0,1,...,h+1 ; j=0,1,...,p-1\}$, where $q_r$ is the all-rejecting state.
\item $F'=\{q_{0,0}\}$.
\item $\forall \sigma\in \Sigma$, $\delta'(q_{i,j},\sigma)=
\begin{cases}
q_r &\text{if } \delta(q_i,\sigma)=q_r,\\
q_{i', (j+1)\bmod p} &\text{if } \delta(q_i,\sigma)=q_{i'}\not = q_r.
\end{cases}$
\end{itemize}
It can be easily verified that the following lemma holds.

\begin{lemma}\label{minimal_DFA}
	$\mathcal{D}_2$ is the minimal DFA that recognizes $L(h,p)$.
\end{lemma}

\subsection{$L(h,p)$ can not be recognized by any MO-1QFA, MM-1QFA or multi-letter 1QFA}
	We recall two useful results from \cite{characterizations,golovkins2002probabilistic}.

\begin{definition}[construction forbidden by MM-1QFA]\label{FMM}
	Given a minimal DFA $\mathcal{A}=(Q, \Sigma, \delta, q_0', F)$, we define \emph{construction forbidden by MM-1QFA} as: $q_1',q_2'$ are distinct states in $Q$ and $\exists x,y\in \Sigma^*$ such that $\delta(q_1',x)=\delta(q_2',x)=q_2'$ and $\delta(q_2',y)=q_1'$.
\end{definition}

\begin{lemma}[\cite{characterizations,golovkins2002probabilistic}]\label{lemma:FMM}
	A regular language $L$ can not be recognized by any MM-1QFA with isolated cut-point if its minimal DFA contains construction forbidden by MM-1QFA.
\end{lemma}

	We have the following proposition.

\begin{proposition}
	Neither MO-1QFA nor MM-1QFA can recognize $L(h,p)$ with isolated cut-point.
\end{proposition}
\begin{proof}
	By \totalref{Lemma}{minimal_DFA}, we know DFA $\mathcal{D}_2$ is the minimal DFA that recognizes $L(h,p)$. In \totalref{Definition}{FMM}, take $q_1'=q_{0,0},q_2'=q_{1,0},x=0^p$ and $y=10^h1^l$, where $l\in\mathbb{Z}^+$ satisfies $|y|\equiv 0 \pmod{p}$. It can be verified that $\mathcal{D}_2$ contains construction forbidden by MM-1QFA. Hence any MM-1QFA can not recognize $L(h,p)$ with isolated cut-point (by \totalref{Lemma}{lemma:FMM}). Since any language recognized by MO-1QFA with isolated cut-point can also be recognized by MM-1QFA with isolated cut-point \cite{strengths}, the claimed result follows.
\end{proof}

	We then recall useful definitions and results from \cite{belovs2007multi}.

\begin{definition}[F-construction]\label{FML}
	Given a DFA $\mathcal{A}=(Q, \Sigma, \delta, q_0, F)$, define \emph{F-construction} as: $q_1',q_2'$ are distinct states in $Q$ and $\exists x,y\in \Sigma^+$ such that $\delta(q_1',x)=\delta(q_2',x)=q_2',\delta(q_1',y)=q_1'$ and $\delta(q_2',y)=q_2'$.
\end{definition}

\begin{lemma}[\cite{belovs2007multi}]\label{lemma:FML}
	A regular language $L$ can be recognized by a multi-letter 1QFA with isolated cut-point iff its minimal DFA does not contain F-construction.
\end{lemma}
	
	Similarly, we have the following proposition.

\begin{proposition}
	No multi-letter 1QFA can recognize $L(h,p)$ with isolated cut-point.
\end{proposition}
\begin{proof}
	By \totalref{Lemma}{minimal_DFA}, we know DFA $\mathcal{D}_2$ is the minimal DFA that recognizes $L(h,p)$. In \totalref{Definition}{FML}, taking $q_1'=q_{0,0},q_2'=q_{1,0},x=0^p$ and $y=(10^h10^h)^p$, we know that $\mathcal{D}_2$ contains F-construction. Hence, by \totalref{Lemma}{lemma:FML}, the proposition holds.
\end{proof}

\subsection{Any PFA recognizing $L(h,p)$ has at least $p$ states}
\begin{lemma}[\cite{strengths,mereghetti2001note}]\label{PFA_p_states}
	Let $p'$ be a prime number. Then any PFA recognizing unary language $L(p')=\{0^{kp'}:k=0,1,\dots\}$ with isolated cut-point has at least $p'$ states.
\end{lemma}

\begin{lemma}
	Any PFA recognizing $L(h,p)$ has at least $p$ states.
\end{lemma}
\begin{proof}
	Assume that there exists a PFA $\mathcal{P}_1=(S,\{0,1\},M_1,\rho,F)$ recognizing $L(h,p)$ with isolated cut-point having fewer than $p$ states. Take $l\in \mathbb{Z}^+$ such that $|0^l10^h| \bmod p\not=0 $, then we get $p|d\Leftrightarrow (0^l10^h)^d\in L(h,p)$. Let  PFA $\mathcal{P}_2=(S,\{0\},M_2,\rho,F)$, where $M_2(0)=M_1(0)^lM_1(1)M_1(0)^h$. We can see that PFA $\mathcal{P}_2$ recognizes $L(p)=\{0^{kp}:k=0,1,\dots\}$ with   isolated cut-point and with fewer than $p$ states. However it contradicts \totalref{Lemma}{PFA_p_states}. Hence, the lemma holds.
\end{proof}

\subsection{Succinctness of 1QFAC}
	In summary, we get a proposition showing the state succinctness of 1QFAC.

\begin{proposition}\label{succinct_1QFAC}
	Let $h$ be any positive integer and let $p$ be any prime number. Then there exists regular languages $L(h,p)$ ($L(h,p)=(1^*00^*10^h)^*\cap \{w:|w|=kp,k=0,1,2,\dots\}$) satisfying
\begin{enumerate}[(I)]
\item Neither MO-1QFA nor MM-1QFA can recognize $L(h,p)$ with isolated cut-point.
\item No multi-letter 1QFA can recognize $L(h,p)$ with isolated cut-point.
\item The minimal DFA recognizing $L(h,p)$ has $(h+2)p+1$ states.
\item Any PFA recognizing $L(h,p)$ with isolated cut-point has at least $p$ states.
\item $\forall \epsilon>0$, there exists a 1QFAC with $h+3$ classical states and $\Theta(\log p)$ quantum basis states recognizing $L(h,p)$ with one-side error $\epsilon$.
 \end{enumerate}
\end{proposition}
\begin{proof}
	 (I)(II)(III)(IV) have been already proved in the previous subsections. The rest is to prove (V). For any $\epsilon>0$, there exists a DFA recognizing $(1^*00^*10^h)$ with $h+3$ states and an MO-1QFA with $\Theta(\log p)$ quantum basis states recognizing $L(p)$ with one-side error $\epsilon$ \cite{strengths}. Hence, by \totalref{Lemma}{1QFAC_construction}, (V) holds.
\end{proof}

\begin{remark}
	By (III) and (V) of \totalref{Proposition}{succinct_1QFAC}, it holds that the bound in \totalref{Theorem}{1QFAC_size} is attainable.
\end{remark}

	When $h$ is small and $p$ is large enough, \totalref{Proposition}{succinct_1QFAC} shows that 1QFAC are exponentially more succinct than DFA and PFA for recognizing some regular languages that can not be recognized by MO-1QFA, MM-1QFA or multi-letter 1QFA.

\section{Simulation}\label{sec:simulation}

	In this section, we show that some 1QFAC can be simulated by MO-1QFA, and multi-letter 1QFA can be simulated by 1QFAC. Then we induce a result regarding a quantitative relationship of the state number between multi-letter 1QFA and DFA.

First we show that any 1QFAC whose DFA is reversible  can be simulated by  MO-1QFA.

\begin{theorem}
	Let 1QFAC $\mathcal{A}_{1QFAC}=(S, Q, \Sigma,\Gamma, s_0, |\psi_0 \rangle, \delta , \mathbb{U}, \mathcal{M} )$  satisfy that $\forall s_1,s_2\in S, \forall \sigma\in\Sigma, \delta(s_1,\sigma)\not=\delta(s_2,\sigma)$ if $s_1\not=s_2$. Then there exists an MO-1QFA simulating $\mathcal{A}_{1QFAC}$ with $|S||Q|$ quantum basis states.
\end{theorem}
\begin{proof}
	Let $ \mathcal{M}=\{\mathcal{M}_s\}_{s\in S}, \mathcal{M}_s=\{P_{s,\gamma}\}_{\gamma\in \Gamma},\mathbb{U}=\{U_{s\sigma}\}_{s\in S,\sigma \in \Sigma}$ and define MO-1QFA $\mathcal{A}_{MO}=(Q, \Sigma,|s_0\rangle |\psi_0 \rangle,\\  \{U'_{\sigma}\}_{\sigma \in \Sigma}, Q_{acc})$ where $U'_{\sigma}=\sum_{s\in S}|\delta(s,\sigma)\rangle\langle s|\otimes U_{s\sigma}$ and the projector onto the subspace generated by $Q_{acc}$ is $P_{acc}=\sum_{a\in Q_{acc}}|a\rangle \langle a|=\sum_s |s\rangle\langle s|\otimes P_{s,acc}$. It can be verified that $U'_{\sigma}$ is unitary and for any $x=x_1x_2\dots x_n$, $U'_{x_n}\cdots U'_{x_2}U'_{x_1}|s_0\rangle |\psi_0 \rangle=|s_x\rangle |\psi_x \rangle$ and ${\rm Prob}_{\mathcal{A}_{MO},acc}(x)={\rm Prob}_{\mathcal{A}_{1QFAC},acc}(x)$. Consequently, the theorem holds.
\end{proof}

Next we show that any $k$-letter 1QFA can be simulated by 1QFAC.

\begin{theorem}\label{ML-1QFA_simulation}
	Let $k$-letter 1QFA $\mathcal{A}_{k-letter}=(Q, \Sigma, |\psi_0\rangle, \{U_{w}\}_{w \in (\{\Lambda\}\cup \Sigma)^k}, Q_{acc})$. Then there exists a 1QFAC simulating $\mathcal{A}_{k-letter}$ with $\sum_{i=0}^{k-1}|\Sigma|^{i}$ classical states and $|Q|$ quantum basis states.
\end{theorem}
\begin{proof}
	The case that $k=1$ is easy. We discuss $k>1$. Let 1QFAC be $\mathcal{A}_{1QFAC}=(S, Q, \Sigma,\Gamma, s_{\Lambda^{k-1}}, |\psi_0 \rangle, \delta , \mathbb{U}, \mathcal{M} )$ and $\mathcal{M}=\{\mathcal{M}_s\}_{s\in S}, \mathcal{M}_s=\{P_{s,\gamma}\}_{\gamma\in \Gamma},\mathbb{U}=\{U_{s\sigma}\}_{s\in S,\sigma \in \Sigma}$, where
\begin{itemize}
\item $S=\{s_w:w\in (\{\Lambda\}\cup \Sigma)^{k-1}$ and $w$ take all possible strings that form by the last $k-1$ letters received by $\mathcal{A}_{k-letter} \}$. We have $|S|= \sum_{i=0}^{k-1}|\Sigma|^{i}$.
\item $\forall \sigma\in\Sigma, w=w_1\dots w_{k-1},\delta(s_w,\sigma)=
\begin{cases}
s_{w_2\dots w_{k-1}\sigma} & k>2,\\
s_{\sigma} & k=2.
\end{cases}$
\item $\forall s_w\in S, U_{s_w\sigma}=U_{w\sigma}$.
\item $\forall s\in S,P_{s,acc}=\sum_{a\in Q_{acc}}|a\rangle \langle a|$.
\end{itemize}
We show that for any input string $x$, the final quantum state of $\mathcal{A}_{k-letter}$ is the same as that of $\mathcal{A}_{1QFAC}$. Let $x=\sigma_1\sigma_2\cdots\sigma_n$. Assume $n\geq k>2$ (other cases can be proved similarly), the final quantum state of $\mathcal{A}_{1QFAC}$ is
\begin{align}
|\psi_x^{1QFAC} \rangle&=U_{\delta(s_{\Lambda^{k-1}},\sigma_1\cdots \sigma_{n-1}) \sigma_n}\cdots U_{\delta(s_{\Lambda^{k-1}},\sigma_1) \sigma_2}U_{s_{\Lambda^{k-1}} \sigma_1}|\psi_0 \rangle\\
&=U_{\sigma_{n-k+1}\sigma_{n-k+2}\cdots \sigma_n}\cdots U_{\Lambda^{k-2}\sigma_1\sigma_2}U_{\Lambda^{k-1}\sigma_1}|\psi_0\rangle.
\end{align}
By equation (\ref{final_state_ML-1QFA}), $|\psi_x^{1QFAC}\rangle$ is equal to the final quantum state of $\mathcal{A}_{k-letter}$ receiving $x$. Hence the theorem holds.

\end{proof}

To conclude this section, we present a result concerning the  relationship of state number between multi-letter 1QFA and DFA.
\begin{theorem}
	If a $k$-letter 1QFA whose alphabet is $\Sigma$ recognizes regular language $L$ with cut-point $\lambda$ isolated by $\epsilon$ and with $n$ quantum basis states, then the minimal DFA of $L$ has $m$ states, where $m\leq (\sum_{i=0}^{k-1}|\Sigma|^{i})(1+\frac{2}{\epsilon})^{2n}$.
\end{theorem}
\begin{proof}
	Immediate from \totalref{Theorem}{1QFAC_size} and \totalref{Theorem}{ML-1QFA_simulation}.
\end{proof}

\section{Conclusions}
Quantum-classical computing models and algorithms have been thought of as important subjects in future study of quantum computing from theory to physical realization due to the difficulty of realizing large-scale universal quantum computers nowadays.  	
1QFAC are a kind of quantum-classical hybrid models that contain quantum and classical states interacting to each other. State complexity in finite automata is an attractive and practical area, so, in this paper, we have centered on studying the state complexity problems concerning 1QFAC, and clarifying the essential relationships between quantum states and classical states in 1QFAC.

 We have proved with improvement the basic relationships between the state number of DFA and 1QFAC for any given regular language. In particular,  we have clarified the trade-offs between the quantum basis states and classical states of 1QFAC and given a lower bound of the classical state number of 1QFAC for recognizing any given regular language. In a way, these results also have shown the superiorities and the limitations of 1QFAC. In addition, we have constructed a regular language showing that 1QFAC does have state advantages over other one-way finite automata such as DFA and PFA, where this language  can not be recognized by MO-1QFA, MM-1QFA or multi-letter 1QFA.

  Finally, we have verified the simulation between 1QFAC, MO-1QFA and multi-letter 1QFA, and showed a quantitative relationship of the state number between multi-letter 1QFA and DFA. However, there are still open problems concerning 1QFAC that are worthy of further consideration:
\begin{itemize}
\item We have given a lower bound of the classical state number of 1QFAC for recognizing any given regular language, but this bound is still not tight if the given language is infinite language, so a natural problem is what is the tight lower bound if the given language is infinite?

\item We say language $L$ is recognized by 1QFAC with cut-point if ${\rm{Prob}}_{acc}(x)> \lambda$ for $x\in L$ and ${\rm{Prob}}_{acc}(x)\leq \lambda$ for $x\not\in L$ for some $\lambda$ in $(0,1)$. Then how to characterize  the language class recognized by 1QFAC with cut-point?
\end{itemize}

\section*{Acknowledgements}
This work is supported in part by the National Natural Science Foundation of China (Nos. 61876195, 61572532) and the Natural Science Foundation of Guangdong Province of China (No. 2017B030311011).	

\end{document}